\documentclass[a4paper,UKenglish]{lipics-v2018}
\hideLIPIcs 

\usepackage{microtype}
\usepackage{microtype}
\usepackage{xspace}
\usepackage{float}

\usepackage{amsmath}
\usepackage{enumitem}

\usepackage{graphicx}
\usepackage{amsfonts}
\usepackage{amsmath}
\usepackage{hyperref}

\newcommand{\M}{\mathbb{M}}
\newcommand{\comp}{\mathrm{comp}}
\renewcommand{\S}{\EuScript{S}}
\newcommand{\R}{\EuScript{R}}
\newcommand{\ms}{M^*}
\newcommand{\opt}{\textsc{Opt}}
\newcommand{\dist}{d}

\newcommand{\wt}{w}
\newcommand{\C}{\EuScript{C}}
\newcommand{\ignore}[1]{}

\newcommand{\Mk}{\EuScript{M}}
\newcommand{\ints}{\mathscr{L}}
\newcommand{\eps}{\varepsilon}

\newcommand{\dir}[1]{\overrightarrow{#1}}

\newcommand{\level}{lev}
\newcommand{\ym}[1]{y_{max}^{#1}}

\newcommand{\sr}{\textrm{csr}}

\newcommand{\appr}{1+\frac{1}{32t}}
\newcommand{\evol}{\EuScript{E}}

\title{Optimal Analysis of an Online Algorithm for the  Bipartite Matching Problem on a Line}

\titlerunning{Online Matching on a Line}

\author{Sharath Raghvendra}{Virginia Tech\\{Blacksburg, USA}}{sharathr@vt.edu}{}{}

\authorrunning{S. Raghvendra}

\Copyright{Sharath Raghvendra}

\subjclass{Theory of Computation $\rightarrow$ Design and Analysis of Algorithms $\rightarrow$ Online Algorithms}


\category{}

\relatedversion{}

\supplement{}

\funding{This work is supported by a NSF CRII grant NSF-CCF 1464276}

\EventEditors{Bettina Speckmann and Csaba D. T{\'o}th}
\EventNoEds{2}
\EventLongTitle{34th International Symposium on Computational Geometry (SoCG 2018)}
\EventShortTitle{SoCG 2018}
\EventAcronym{SoCG}
\EventYear{2018}
\EventDate{June 11--14, 2018}
\EventLocation{Budapest, Hungary}
\EventLogo{socg-logo} 
\SeriesVolume{99}
\ArticleNo{67} 
\nolinenumbers

\begin{document}

\maketitle
\begin{abstract}
In the online metric bipartite matching problem, we are given a set $S$ of server locations in a metric space. Requests arrive one at a time, and on its arrival, we need to immediately and irrevocably match it to a server at a cost which is equal to the distance between these locations. A $\alpha$-competitive algorithm will assign requests to servers so that the total cost is at most $\alpha$ times the cost of $M_{\opt}$ where $M_{\opt}$ is the minimum cost matching between $S$ and $R$. 

We consider this problem in the adversarial model for the case where $S$ and $R$ are points on a line and $|S|=|R|=n$. We improve the analysis of the deterministic Robust Matching Algorithm (RM-Algorithm, Nayyar and Raghvendra FOCS'17) from $O(\log^2 n)$ to an optimal $\Theta(\log n)$. Previously, only a randomized algorithm under a weaker oblivious adversary  achieved a competitive ratio of $O(\log n)$ (Gupta and Lewi, ICALP'12). The well-known Work Function Algorithm (WFA) has a competitive ratio of $O(n)$ and $\Omega(\log n)$ for this problem. Therefore,  WFA cannot achieve an asymptotically better competitive ratio than the RM-Algorithm.
 \end{abstract}

\section{Introduction}

Driven by consumers' demand for  quick access to goods and services, business ventures schedule their delivery  in real-time, often without the complete knowledge of the future request locations or their order of arrival. Due to this lack of complete information, decisions made tend to be sub-optimal. Therefore, there is a need for competitive \emph{online algorithms} which immediately and irrevocably allocate resources to requests in real-time by incurring a small cost.

Motivated by these real-time delivery problems, we study the problem of computing the online metric bipartite matching of requests to servers.  Consider servers placed in a metric space where each server has a capacity that restricts how many requests it can serve. When a new request arrives, one of the servers with positive capacity is matched to this request. After this request is served, the capacity of the server reduces by one. We assume that the cost associated with this assignment is a metric cost; for instance, it could be the minimum distance traveled by the server to reach the request.

The case where the capacity of every server is $\infty$ is the celebrated $k$-server problem. The case where every server has a capacity of $1$ is the \emph{online metric bipartite matching problem}. In this case, the requests arrive one at a time, we have to immediately and irrevocably match it to some unmatched server. The resulting assignment is a matching and is referred to as an \emph{online matching}. An optimal assignment is impossible since an adversary can easily fill up the remaining locations of requests in $R$ in a way that our current assignment becomes sub-optimal. Therefore, we want our algorithm to compute an assignment that is competitive with respect to the optimal matching. For any input $S, R$ and any arrival order of requests in $R$, we say our algorithm is \emph{$\alpha$-competitive}, for $\alpha > 1$, when the cost of the online matching $M$ is at most $\alpha$ times the minimum cost, i.e., 
$$w(M) \le \alpha w(M_{\opt}).$$
Here $M_{\opt}$ is the minimum-cost matching of the locations in $S$ and $R$.
In the above discussion, note the role of the adversary. In the \emph{adversarial model}, the adversary knows the server locations and the assignments made by the algorithm and generates a sequence to maximize $\alpha$. In the weaker oblivious adversary model, the adversary knows the randomized algorithm but does not know the  random choices made by the algorithm. In this paper, we consider the online metric bipartite matching problem in the adversarial model and where  $S$ and $R$ are points on a line.

Consider the adversarial model. For any request, the \emph{greedy heuristic} simply assigns the closest unmatched server to it. The greedy heuristic, even for the line metric, is only $2^n$-competitive~\cite{gl_icalp12} for the online matching problem.     The well-known Work Function Algorithm  (WFA) chooses a server that minimizes the sum of the \emph{greedy cost} and the so-called \emph{retrospective cost}. For the $k$-server problem, the competitive ratio of the WFA\ is $2k-1~\cite{kp_jacm}$ which is near optimal with a lower bound of $k$ on the competitive ratio of any algorithm in any metric space that has at least $k+1$ points~\cite{mmd_joa}. 

In the context of online metric matching problem,  there are algorithms that achieve a  competitive ratio of $2n-1$ in the adversarial model~\cite{r_approx16, kmv_tcs, kp_jalg}.
This competitive ratio is worst-case optimal, i.e., there exists a metric space where we cannot do better. However, for Euclidean metric, for a long time,  there was a stark contrast between the upper bound of $2n-1$ and the lower bound of $\Omega(n^{1-1/d })$. Consequently, significant effort has been made to study the performance of online algorithms in special metric spaces, especially the  line metric. For example, for the line metric, it has been shown that the WFA when applied to the online matching problem has a lower bound of $\Omega(\log n)$ and an upper bound of $O(n)$~\cite{kn2004}; see also~\cite{Antoniadis2015} for a $O(n^{0.585})$competitive algorithm for the line metric. In the oblivious adversary model, there is an $O(\log n)$-competitive~\cite{gl_icalp12} algorithm for the line metric.  There is also an $O(d\log n)$-competitive algorithm in the oblivious adversary model for any metric space with doubling dimension $d$~\cite{gl_icalp12}. 

Only recently, for any metric space and for the adversarial model, Raghvendra and Nayyar~\cite{nr_focs17} provided a bound of $O(\mu(S)\log ^2 n)  $ on the competitive ratio of the RM-Algorithm  -- an algorithm that was  introduced by Raghvendra~\cite{r_approx16}; here $\mu(S)$ is the worst case ratio of the the TSP and diameter of any positive diameter subset of $S$.  There is a simple lower bound of $\Omega(\mu(S))$ on the competitive ratio of any algorithm for this problem. Therefore, RM-Algorithm is near-optimal for every input set $S$. When $S$ is a set of points on a line, $\mu(S) = 2$ and therefore, their analysis bounds the competitive ratio of the RM-Algorithm by $O(\log^2 n)$ for the line metric and $O(n^{1-1/d}\log^2 n)$ for any $d$-dimensional Euclidean metric.   Furthermore,  RM-Algorithm also has a lower bound of $\Omega(\log n)$ on its competitive ratio for the line metric. In this paper, we provide a different analysis  and show that the RM-Algorithm is $\Theta(\log n)$-competitive.  

\subparagraph*{Overview of RM-Algorithm:} At a high-level, the RM-Algorithm maintains two matchings $M$ and $M^*$, both of which match requests seen so far to the same subset of servers in $S$. We refer to $M$\ as the online matching and $M^*$ as the offline matching. For a parameter $t > 1$ chosen by the algorithm, the offline matching $M^*$ is a $t$-approximate minimum-cost matching satisfying a set of  relaxed dual feasibility conditions of the  linear program for the minimum-cost matching; here each constraint relaxed by a multiplicative factor $t$. Note that, when $t=1$, the matching $M^*$ is simply the minimum-cost matching.  

 When the $i^{th}$  request $r_i$ arrives, the algorithm  computes an augmenting path $P_{i}$ with the minimum ``cost'' for an appropriately defined cost function. This path $P_i$  starts at $r_i$ and ends at an unmatched server $s_i$.  The algorithm then augments the offline matching $M^*$ along $P$ whereas the online matching $M$ will simply match the two  endpoints of $P_{i}$. Note that $M$ and $M^*$ will always match requests  to the same subset of servers. 
We refer to the steps taken by the algorithm to process request $r_i$  as the $i^{th}$ phase of the algorithm. For $t >1$, it has been shown in~\cite{r_approx16} that the sum of the costs of every augmenting path computed by the algorithm is bounds  the online matching cost from above.  Nayyar and Raghvendra~\cite{nr_focs17}  use this property and bounded the ratio of the sum of the costs of augmenting paths to the cost of the optimal matching. In order to accomplish this they associate every request $r_i$ to  the cost of $P_i$. To bound the sum of costs of all the augmenting paths, they  partition the requests into   $\Theta(\log^2 n)$ groups and within each group they bound this ratio  by $\mu(S)$ (which is a constant when $S$ is a set of points on a line). For the line metric, each group can, in the worst-case, have a ratio of $\Theta(1)$. However,  not all groups can simultaneously exhibit this worst-case behavior. In order to improve the competitive ratio from  $O(\log^2 n)$ to $O(\log n)$, therefore, one has to bound the combined ratios of several groups by a constant making the analysis challenging. 
\subsection{Our Results and Techniques}
In this paper, we show that when the points in $S\cup R$ are  on a line, the RM-Algorithm achieves a competitive ratio of $O(\log n)$.  Our analysis is tight as there is an  example in the line metric for which the RM-Algorithm produces an online matching which is $\Theta(\log n)$ times the  optimal cost. We achieve this improvement using the following new ideas: 
\begin{itemize}
\item  First, we establish the the \emph{ANFS-property} of the RM-Algorithm (Section~\ref{sec:nn}). We show that many requests are matched to an approximately nearest free server (ANFS) by the algorithm. We define certain edges of the online matching $M$ as \emph{short} edges and show that every short edge matches the request to an approximately nearest free server. Let $M_H$ be the set of short edges of the online matching and $M_L=M\setminus M_{H}$ be the \emph{long} edges. We also show that when $t=3$, the total cost of the short edges $w(M_H)\ge \wt(M)/6$ and therefore, the cost of the long edges is $\wt(M_L) < (5/6)\wt(M)$.   

\item  For every point in $S\cup R$, the RM-Algorithm maintains a dual weight (Section~\ref{sec:algorithm}). For  our analysis in the line metric, we assign an interval to every request. The length of this interval  is determined by the dual weight of the request. At the start of phase $i$,  let $\sigma_{i-1}$ be the union of all such intervals. By its construction, $\sigma_{i-1}$ will consist of a set of interior-disjoint intervals. While processing request $r_i$, the RM-Algorithm conducts a series of dual adjustments and a  subset of requests (referred to as $B_i$) undergo an increase in dual weights. After the dual adjustments, the union of the intervals for requests in $B_i$ forms a single interval and these increases conclude in the discovery of the minimum cost augmenting path.  Therefore, after phase $i$,  intervals in $\sigma_{i-1}$ may grow  and combine to form a single interval $\EuScript{I}$ in $\sigma_i$. Furthermore, the new online edge $(s_i,r_i)$ is also contained inside this newly formed interval $\EuScript{I}$ (Section~\ref{sec:dual}). Based on the length of the interval $\EuScript{I}$, we assign one of $O(\log n)$ levels to the edge $(s_i,r_i)$. This partitions all the online edges in $O(\log n)$ levels. 

\item\ The online edges of any given level $k$ can be expressed as several non-overlapping well-aligned matching of a well separated input (Section~\ref{sec:1danalysis}) . We establish properties of such matchings (Section~\ref{sec:wam} and Figure~\ref{fig:wam}) and use it to  bound the total cost of the ``short'' online edges of level $k$ by the sum of $\wt(M_{\opt})$ and $\gamma$ times the cost of the long edges of level $k$, where $\gamma$ is a small positive constant (Section~\ref{sec:1danalysis}). Adding across the $O(\log n)$ levels, we get $\wt(M_H) \le O(\log n)\wt(M_{\opt}) + \gamma\wt(M_L)$. Using the ANFS-property of short and long edges, we immediately get $(1/6-5\gamma/6)\wt(M) \le O(\log n)\wt(M_{\opt})$. For a sufficiently small $\gamma$, we bound the competitive ratio, i.e., $\wt(M)/\wt(M_{\opt})$ by $O(\log n)$.  
  
\end{itemize}

\subparagraph*{Organization:}   
The rest of the paper is organized as follows. We begin by presenting (in Section~\ref{sec:algorithm}) the RM-Algorithm and some of its use properties as shown in~\cite{r_approx16}. For the analysis, we establish the \emph{ANFS-property} in Section~\ref{sec:nn}. After that, we will (in Section~\ref{sec:wam}) introduce well aligned matchings of well-separated inputs on a line.  Then, in Section~\ref{sec:dual}, we interpret the dual weight maintained for each request as an interval and study the properties of the union of these intervals.  Using these properties (in Section~\ref{sec:1danalysis}) along with the $\textrm{ANFS}$-property of the algorithm, we will  establish a bound of $O(\log n)$ on the competitive ratio for the line metric. 

\section{Background and Algorithm Details}
\label{sec:algorithm}
In this section, we introduce the relevant background and describe the $\textrm{RM}$-algorithm.

A \emph{matching} $M\subseteq S\times R$ is any set of vertex-disjoint edges of the complete bipartite graph denoted by $G(S,R)$. The cost of any edge $(s,r) \in S\times R$ is given by $d(s,r)$; we assume that the cost function satisfies the metric property.\ The cost of any matching $M$ is given by the sum of the costs of its edges, i.e., $\wt(M) = \sum_{(s,r)\in M}d(s,r)$. A \emph{perfect matching} is a matching where every server in $S$ is serving exactly one request in $R$, i.e., $|M|=n$. A \emph{minimum-cost perfect matching} is a perfect matching with the smallest cost.

Given a matching $\ms$, an \emph{alternating path} (resp. cycle) is a simple path (resp. cycle) whose edges alternate between those in $\ms$ and those not in $\ms$. We refer to any vertex that is not matched in $\ms$ as a \emph{free vertex}. An \emph{augmenting path} $P$ is an alternating path between two free vertices. We can \emph{augment} $\ms$ by one edge along $P$ if we remove the edges of $P \cap \ms$ from $\ms$ and add the edges of $P \setminus \ms$ to $\ms$. After augmenting, the new matching is precisely given by $\ms \oplus P$, where $\oplus$ is the symmetric difference operator. A matching $\ms$ and a set of  dual weights, denoted by $y(v)$  for each point  $v\in S\cup R$, is a \emph{$t$-feasible matching} if,
for any request $r \in R$ and a server $s\in S$,  the following conditions hold:
\begin{eqnarray}
y(s) + y(r) &\le&  td(s,r), \label{eq:feas1}\\
y(s) + y(r)  &=&  d(s,r)\quad \mbox{if } (s,r) \in \ms. 
\label{eq:feas2}
\end{eqnarray}

  Also, we refer to an edge $(s,r) \in S\times R$ to be \emph{eligible} if either $(s,r) \in \ms$ or $(s,r)$ satisfies inequality~\eqref{eq:feas1} with equality:
\begin{eqnarray}
y(s) + y(r) &=&  td(s,r), \quad \mbox{if } (s,r) \notin \ms \label{eq:elig1}\\
y(s) + y(r)  &=&  d(s,r)\quad \mbox{if } (s,r) \in \ms. 
\label{eq:elig2}
\end{eqnarray}

For a parameter $t \ge 1$, we define the \emph{$t$-net-cost} of any augmenting path $P$ with respect to $\ms$ to be:
$$\phi_t(P)=t\left(\sum_{(s,r)\in P\setminus \ms}d(s,r)\right) -\sum_{(s,r)\in P\cap \ms}d(s,r).$$
The definitions of $t$-feasible matching, eligible edges and $t$-net cost (when $t=1$) are also used in describing the well-known Hungarian algorithm which computes the minimum-cost matching.  In the Hungarian algorithm, initially $\ms = \emptyset$ is a $1$-feasible matching with all the dual weights $y(v)$  set to $0$. In each iteration, the Hungarian Search procedure adjusts the dual weights $y(\cdot)$ and computes an augmenting path $P$ of eligible edges while maintaining the $1$-feasibility of $\ms$ and then augments $\ms$ along $P$. The augmenting path $P$ computed by the standard implementation of the Hungarian search procedure can  be shown to also have the minimum $1$-net-cost.  

Using this background, we  describe the $\textrm{RM}$-Algorithm. At the start, the value of   $t    $ is chosen at the start of the algorithm.  The algorithm maintains two matchings: an online matching $M$ and a $t$-feasible matching (also called the \emph{offline matching}) $\ms$ both of which are initialized to $\emptyset$. After processing $i-1$ requests, both matchings $M$ and $\ms$  match each of the $i-1$ requests to the same subset of servers in $S$, i.e., the set of free (unmatched) servers $S_F$ is the same for both   $M$ and  $\ms$.  To process the $i^{th}$ request $r_i$, the algorithm does the following
\begin{enumerate}
\item 
Compute the minimum $t$-net-cost augmenting path $P_i$ with respect to the offline matching $\ms$. Let $P_i$ be this path starting from $r_i$ and ending at some server $s_i \in S_F$.
\item Update offline matching $\ms$ by augmenting it along $P_i$, i.e., $\ms\leftarrow\ms\oplus P_i$ and update online matching $M$ by matching $r_i$ to $s_i$. $M \leftarrow M\cup\{(s_i,r_i)\} $. 
\end{enumerate}
  While computing the minimum $t$-net-cost augmenting path in Step $1$, if there are multiple paths with the same $t$-net-cost, the algorithm will simply select the one with the fewest number of edges. Throughout this paper, we set $t=3$. In~\cite{r_approx16}, we present an $O(n^2)$-time  search algorithm that is similar to the Hungarian Search to compute the minimum $t$-net-cost path in Step 1 any phase $i$ and we describe this next. 

   The implementation of Step $1$ of the algorithm is similar in style to the Hungarian Search procedure. To  compute the minimum $t$-net-cost path $P_i$ , the algorithm grows an alternating tree consisting only of eligible edges. There is an alternating path of eligible edges from $r_i$ to every server and request participating in this tree. To grow this tree, the algorithm increases the dual weights of every request in this tree until at least one more edge becomes eligible and a new vertex enters the tree. In order to maintain feasibility, the algorithm reduces the dual weights of all the servers in this tree by the same amount. This search procedure ends when an augmenting path $P_i$ consisting only of eligible edges is found. Let $B_i$ (resp. $A_i$) be the set of requests (resp. servers) that participated in the alternating tree of phase $i$. Note that during Step $1$, the dual weights of requests in $B_i$ may only increase and the dual weights of servers in $A_i$ may only reduce.

The second step begins once the augmenting path $P_i$ is found. The algorithm augments the offline matching $\ms$ along this path. Note that, for the $\ms$ to be $t$-feasible, the edges that newly enter  $\ms$ must satisfy (\ref{eq:feas2}). In order to ensure this, the algorithm will reduce the dual weight of each request $r$ on $P_i$  to $y(r) \leftarrow y(r)-(t-1)d(s,r)$. Further details of the algorithm and proof of its correctness can be found in~\cite{r_approx16}. In addition, it has also been shown that the algorithm maintains the following three invariants:
\begin{enumerate}[label=(I\arabic*)]
\item The offline matching $\ms$ and dual weights $y(\cdot)$ form a $t$-feasible matching,
\item For every server $s \in S$, $y(s) \le 0$ and if $s \in S_F$, $y(s)=0$. For every request $r \in R$, $y(r)\ge 0$ and if $r$  has not yet arrived, $y(r)=0$,
\item At the end of the first step of phase $i$ of the algorithm the augmenting path $P_i$ is found and the dual weight of $r_i$, $y(r_i)$, is equal to the $t$-net-cost $\phi_t(P_i)$. 
\end{enumerate}

\subparagraph*{Notations:}Throughout the rest of this paper, we will use the following notations. We will index the requests in their order of  arrival, i.e.,  $r_i$ is the $i$th request to arrive. Let $R_i$ be the set of first $i$ request. Our algorithm processes the request $r_i$, by computing an augmenting path $P_i$. Let $s_i$ be the free server at the other end of the augmenting path $P_i$. Let $\mathbb{P}=\{P_1,\ldots,P_n\}$ be the set of $n$ augmenting paths generated by the algorithm.  In order to compute the augmenting path $P_i$, in the first step, the algorithm adjusts the dual weights and constructs an alternating tree; let $B_i$ be the set of requests and let $A_i$\ be the set of servers that participate in this alternating tree. Let 
$\ms_i$ be the offline matching after the $i$th request has been processed; i.e., the matching obtained after augmenting the matching $\ms_{i-1}$ along $P_i$.  Note that $\ms_0=\emptyset$  and $\ms_n = \ms$ is the final matching after all the $n$ requests have been processed. The online matching $M_i$ is the online matching after $i$ requests have been processed. $M_i$ consists of edges $\bigcup_{j=1}^i (s_j,r_j)$. Let $S^i_F$ be the free servers with respect to matchings $M_{i-1}$ and $\ms_{i-1}$, i.e., the set of free servers at the start of phase $i$.
For any path $P$, let $\ell(P) = \sum_{(s,r)\in P}d(s,r)$ be its \emph{length}. 

Next, in Section~\ref{sec:nn}, we will present the approximate nearest free server (ANFS) property of the RM-Algorithm. In Section~\ref{sec:wam}, we present an well aligned matching of a well separated input instance. In Section~\ref{sec:dual},  we interpret the execution of each phase of the $\textrm{RM}$-Algorithm in the line metric. Finally, in Section~\ref{sec:1danalysis}, we  give our analysis of the algorithm for the line metric.

\section{New Properties of the Algorithm}
In this section, we present new properties of the $\textrm{RM}$-Algorithm. First, we show that the RM-Algorithm will assign an approximate nearest free server for many requests and show that the total cost of these ``short'' matches will be at least one-sixth of the online matching cost.  Our proof of this property is valid for any metric space.

\subsection{Approximate Nearest Free Server Property}
\label{sec:nn}

We divide the $n$ augmenting paths $\mathbb{P}=\{P_1,\ldots, P_n\}$ computed by the RM-Algorithm into two sets, namely \emph{short} and \emph{long}  paths.  For any $i$, we refer to $P_i$ as a \emph{short} augmenting path if $\ell(P_i) \le \frac{4}{t-1}\phi_t(P_i)$ and \emph{long} otherwise.  Let $H\subseteq \mathbb{P}$ be this set of all short augmenting paths and $L=\mathbb{P}\setminus H$ be the \emph{long} augmenting paths. In phase $i$,  the algorithm adds an edge between $s_i$ and $r_i$ in the online matching.  We refer to any edge of the online matching $(s_i,r_i)$ as a \emph{short edge} if $P_i$ is a short augmenting path. Otherwise, we refer to this edge as a \emph{long} edge. The   set of all short edges,  $M_H$ and the set of long edges  $M_L$ partition the edges of the online matching $M$.
 
At the start of phase $i$,  $S_F^{i}$ are the set of free servers.  Let $s^*\in S_F^{i}$ be the server closest to $r_{i}$, i.e., $s^*=\arg\min_{s\in S_F^{i}}d(r_i,s).$
Any other server $s \in S_F^{i}$ is an \emph{$\alpha$-approximate nearest free server} to the request $r$ if $$d(r_{i},s) \le \alpha\dist(r_{i},s^*).$$

 In Lemma~\ref{lem:anfs} and~\ref{lem:shortcost}, we show that the short edges in $M_{H}$ match a request to a  $6$-ANFS and the cost of $w(M_H)$ is at least one-sixth of the cost of the online matching.  
 
\begin{lemma}
\label{lem:anfs}
For any request $r_{i}$, if $P_i$ is a short augmenting path, then $s_i$ is a $(4+\frac{4}{t-1})$-ANFS of $r_i$.
\end{lemma}
\begin{proof}
Let $s^*$ be the nearest available server of $r_i$ in $S^{i}_F$. Both  $s^*$ and $r_i$ are free and so the edge $P=(s^*,r_i)$ is  also an augmenting path with respect to $\ms_{i-1}$ with $\phi_t(P)=t\dist(s^*,r_i)$. The algorithm computes $P_i$ which is the minimum $t$-net-cost path with respect to $\ms_{i-1}$ and so, 
$$\phi_t(P_i) \le t d(s^*,r_i).$$
Since $P_i$ is a short augmenting path,
\begin{eqnarray}
\frac{(t-1)}{4}d(s_i,r_i) \le \phi_t(P_i) \le td(s^*,r_i),\\
d(s_i,r_i) \le \frac{4t}{t-1}d(s^*,r_i) = (4+ \frac{4}{t-1})d(s^*,r_i),
\end{eqnarray}
implying that $s_i$ is a $(4+\frac{4}{t-1})$-approximate nearest free server to the request $r_i$.
\end{proof}
 
 \begin{lemma}
\label{lem:shortcost}
 Let $M_H$ be the set of short edges of the online matching $M$. Then,
 \begin{equation}
 \label{eq:shortcost}
 (4+\frac{4}{t-1})\wt(M_H) \ge \wt(M).
 \end{equation}  
 \end{lemma}
 \begin{proof}
 Since the matchings $\ms_i$ and $\ms_{i-1}$ differ only in the edges of the augmenting path $P_i$, we have
 \begin{eqnarray}
 \wt(\ms_i) - \wt(\ms_{i-1}) &=& \sum_{(s,r) \in P_i \setminus
 \ms_{i-1}}d(s,r) - \sum_{(s,r) \in P_i\cap \ms_{i-1}}d(s,r)
 \label{eq:dist}\\
 &=& \phi_t(P_i) - \left((t-1)\sum_{(s,r) \in P_i \setminus
 \ms_{i-1}}d(s,r)\right)\nonumber\\
 &=& \phi_t(P_i) - \left(\frac{t-1}{2} \sum_{(s,r) \in P_i \setminus
 \ms_{i-1}}d(s,r) + \frac{t-1}{2} \sum_{(s,r) \in P_i \setminus
 \ms_{i-1}}d(s,r)\right).\nonumber
 \end{eqnarray}

 The second equality follows from the definition of $\phi_t(\cdot)$.
Adding and subtracting $\displaystyle (\frac{t-1}{2}) \sum_{(s,r) \in
 P_i\cap \ms_{i-1}}d(s,r)$ to the RHS we get,

 \begin{eqnarray*}
 \wt(\ms_i) - \wt(\ms_{i-1})
 &=& \phi_t(P_i) - \frac{t-1}{2} \left(\sum_{(s,r) \in P_i \setminus
 \ms_{i-1}}d(s,r) +\sum_{(s,r) \in P_i\cap
 \ms_{i-1}}d(s,r)\right)\\
 & &  - \frac{t-1}{2} \left(\sum_{(s,r) \in P_i \setminus
 \ms_{i-1}}d(s,r) - \sum_{(s,r) \in P_i\cap
 \ms_{i-1}}d(s,r)\right)\\
 &=& \phi_t(P_i) - \frac{t-1}{2}(\sum_{(s,r) \in P_i}d(s,r)) -
 \frac{t-1}{2} ( \wt(\ms_i) - \wt(\ms_{i-1})).
 \end{eqnarray*}

 The last equality follows from~\eqref{eq:dist}. Rearranging terms and
 setting $\sum_{(s,r) \in P_i} d(s,r) = \ell(P_i)$, we get,

 \begin{eqnarray}
 \frac{t+1}{2} (\wt(\ms_i) - \wt(\ms_{i-1})) &=&  \phi_t(P_i) -
 \frac{t-1}{2}\ell(P_i), \nonumber\\
 \frac{t+1}{2} \sum_{i=1}^n(\wt(\ms_i) - \wt(\ms_{i-1}))&=&
 \sum_{i=1}^n\phi_t(P_i) - \frac{t-1}{2}\sum_{i=1}^n \ell(P_i),\nonumber\\
 0 &\le& \sum_{i=1}^n\phi_t(P_i) - \frac{t-1}{2}\sum_{i=1}^n\ell(P_i),\nonumber\\
 \sum_{i=1}^n\phi_t(P_i) &\ge& \frac{t-1}{2}\sum_{i=1}^n\ell(P_i).\label{eq:costlength}
 \end{eqnarray}
 In the second to last equation, the summation on the LHS telescopes
 canceling all terms except $\wt(\ms_n)- \wt(\ms_0)$. Since $\ms_n=\ms$ and
 $\ms_0$ is an empty matching, we get $\wt(\ms_n)- \wt(\ms_0)=\wt(\ms)$. As $\wt(\ms)$ is always a positive value, the second to last equation follows.

 Recollect that $H$ is the set of short augmenting paths and $L$ is the set of long augmenting paths with $\mathbb{P}=L\cup H$.  We rewrite (\ref{eq:costlength})
$$
 \sum_{P_i\in H}\phi_t(P_i) + \sum_{P_i\in L}\phi_t(P_i) \ge \frac{t-1}{2}\sum_{P_i\in H} \ell(P_i) + \frac{t-1}{2}\sum_{P_i\in L} \ell(P_i).$$

\begin{eqnarray}
 \sum_{P_i\in H}\phi_t(P_i) + \sum_{P_i\in L}\phi_t(P_i) &\ge& \frac{t-1}{2}\sum_{P_i\in H} \ell(P_i) +2 \sum_{P_i\in L}\phi_t(P_i),\nonumber\\
 \sum_{P_i\in H}\phi_t(P_i) &\ge& \sum_{P_i\in L}\phi_t(P_i).\label{eq:longshort}
 \end{eqnarray}
 
 The last two inequalities follow from the fact that  $\frac{t-1}{2}\sum_{P_i\in H}\ell(P_i)$ is a positive term and also the  definition of long paths, i.e., if $P_i$ is a long path then $\phi_t(P_i) \le \frac{t-1}{4}\ell(P_i)$. Adding $\sum_{P_i\in H}\phi_t(P_i)$ to~\eqref{eq:longshort} and applying~\eqref{eq:costlength}, we get 
 \begin{equation}
 \label{eq:shortcost1}
 2\sum_{P_i \in H}\phi_t(P_i) \ge \sum_{P_i\in L\cup H}\phi_t(P_i)\ge \frac{t-1}{2}\sum_{i=1}^n \ell(P_i) \ge \frac{t-1}{2}\wt(M).\end{equation}

When request $r_i$ arrives, the edge $P'=(r_i,s_i)$ is an augmenting path of length $1$ with respect to $\ms_{i-1}$and has a $t$-net-cost $\phi_{t}(P') = td(r_i,s_i)$. Since $P_i$ is the minimum $t$-net-cost path, we have
$\phi_t(P_i) \le t\dist(s_i,r_i)$.
Therefore, we can write~\eqref{eq:shortcost1}  as
\begin{eqnarray*}
2t\wt(M_H) = 2t\sum_{P_i\in H} d(s_i,r_i) \ge 2\sum_{P_i\in H}\phi_t(P_i) \ge \frac{t-1}{2}\wt(M),
\end{eqnarray*}
or $(4+\frac{4}{t-1})\wt(M_H)\ge \wt(M)$ as desired.

 \end{proof}
   If we set $t =3$, then $6\wt(M_H) \ge \wt(M)$ or $\wt(M_H) \ge \wt(M)/6$.

\subparagraph*{Convention and Notations for the Line Metric:}  When $S$ and $R$ are points on a line, any point $v \in S\cup R$ is simply a real number. We can interpret any edge $(s,r) \in S\times R$ as the line segment that connects $s$ and $r $ and its cost $d(s,r)$ as $|s-r|$. We will abuse notation and use $(s,r)$ to denote both the edge as well as the corresponding line segment. A matching of $S$ and $R$, therefore, is also a set of line segments, one corresponding to each edge. For any closed interval $I = [x,y]$, $open(I)$ will be the open interval $(x,y)$. We define the boundary, $bd(I)$ of any closed (or open) interval $I=[x,y]$ (or $I=(x,y)$) to be the set of two end points$\{x,y\}$.\     The optimal matching of points on a line has a very simple structure which we describe next.

\subparagraph*{Properties of optimal matching on a line:}
For any point set $K$  on a line, let $\sigma(K)$ be a sequence of points of $K$  sorted in increasing order of their coordinate value.
Given sets $S$ and $R$, consider sequences $\sigma(S)=\langle s^1,\ldots,s^n \rangle$ and $\sigma(R)=\langle r^1,\ldots, r^n\rangle$. The minimum-cost matching $M_{\opt}$ will match server $s^i$ to request $r^i$. We will show this next.  

 \  Let $K\subset S\cup R$.  Suppose $K$ contains $n_1$ points from $S$ and $n_2$ points from $R$ and  let $\mathrm{diff}(K) = |n_1-n_2|$.  Consider $\sigma(S\cup R)=\langle k^1,\ldots, k^{2n}\rangle$ and let $K_j$ be  the first $j$ points in the sequence $\sigma(S\cup R)$. Note that, for any $j$,  $\mathrm{diff}(K_j)=\mathrm{diff}((S\cup R)\setminus K_j)$ and there are precisely $\mathrm{diff}(K_j)$ more (or fewer) vertices of  $S$ than $R$ in $K_j$. Consider the intervals $\kappa^j = [k^j,k^{j+1}]$, for every $1\le j \le 2n-1$ with a length $\ints(\kappa^j)=|k^{j+1}-k^j|$.  Any perfect matching will have at least $\mathrm{diff}(K_j)$ edges with  one end point in $K_j$ and the other in $(S\cup R)\setminus K_j$. Every such edge will contain the interval $\kappa^j$ and so, the cost of any perfect matching $M$ is at least $\wt(M) \ge \sum_{j=1}^{2n-1}\mathrm{diff}(K_j)\ints(\kappa^j)$. 

  We claim that the matching   $M_{\opt}$ described above has a cost  $ \sum_{j=1}^{2n-1}\mathrm{diff}(K_j)\ints(\kappa^j)$ and so $M_{\opt}$ is an optimal matching. For any $j$, without loss of generality, suppose there are $\mathrm{diff}(K_j)$ more points of $S$ than $R$ in $K_j$ and so $K_j$ contains the points $\langle r^1,\ldots, r^{(j-\mathrm{diff}(K_j))/2}\rangle$ and $\langle s^1,\ldots, s^{(j+\mathrm{diff}(K_j))/2}\rangle$. In the optimal solution, we match $\langle r^1,\ldots, r^{(j-\mathrm{diff}(K_j))/2}\rangle$ with $\langle s^1,\ldots, s^{(j-\mathrm{diff}(K_j))/2}\rangle$; the remaining $\mathrm{diff}(K_j)$ servers in $K_j$  match to requests in $(S\cup R)\setminus K_j$ . Therefore, for every $1\le j \le 2n-1$, there are exactly $\mathrm{diff}(K_j)$ edges in $M_{\opt}$ that contain the interval $\kappa^j$ and so, we can express its cost as $\wt(M_{\opt})= \sum_{j} \mathrm{diff}(K_j)\ints(\kappa^j)$.

Any edge $(s,r)$ of the matching is called a \emph{left edge} if the server $s$ is to the left of request $r$. Otherwise, we refer to the edge as the \emph{right edge}. Consider any perfect matching $M$   of $S$ and $R$ and  for an interval $\kappa^j = [k^j , k^{j+1}]$, let $M(\kappa^j)$ be the edges of $M$ that contain the interval $\kappa^j$. For every interval $\kappa^j$,  if  the edges in $M(\kappa^j)$ are either all left edges or all right edges, then as an easy consequence from the above discussion, it follows that $M$ is an optimal matching.    
\begin{enumerate}[label=(OPT)]
\item For every interval $\kappa^j$,  if  the edges in $M_{\kappa^j}$ are either all left edges or all right edges, then   $M$ is an optimal matching.
\end{enumerate}

\subsection{Properties of 1-dimensional Matching}
\label{sec:wam} 
In this section, we define certain input instances that we refer to as a well-separated input instance. We then define matchings that are well-aligned for such instance and then bound the cost of such a well-aligned matching by the cost of their optimal matching. In Section~\ref{sec:1danalysis}, we divide the edges of the online matching into well-aligned matchings on well-separated instances. This will play a critical role in bounding the competitive ratio for the line metric.  
\subparagraph*{Well-separated instances and well aligned matchings:} A \emph{well-separated instance} of the $1$-dimensional matching problem is defined as follows.  
For any $\Delta > 0$ and $0 < \eps \le 1/8$, consider the following four intervals $I_{M}=[0, \Delta]$, $I_L=[-\eps\Delta, 0],  I_R=[\Delta, (1+\eps)\Delta]$ and $I_A=[-\eps\Delta, (1+\eps)\Delta]$.  Note that $I_L$ is the leftmost interval, $I_R$ is the rightmost interval, and $I_M$ lies in the middle of the $I_L$ and $I_R$.  $I_A$ is simply the union of $I_L, I_R$, and $I_M$. We say that any input set of servers $S$ and requests $R$  is an $\eps$-well-separated input if, for some $\Delta > 0$, there is a translation of $S\cup R$ such that $S\subset I_L\cup I_R$ and $R\subset I_A$. 

Given an $\eps$-well-separated input $S$ and $R$,  consider the intervals, $I_L' = [-\eps\Delta, \eps\Delta]$ and $I_R' = [(1-\eps)\Delta, (1+\eps)\Delta]$.  We  divide the edges of any matching $M$ of $S$ and $R$ into three groups. Any edge $(s,r) \in M$ is \emph{close} if $s,r \in I_L'$ or $s,r \in I_R'$.  Any edge $(s,r) \in M$ is a \emph{far} edge if $(s,r) \in I_L'\times I_R'$ or $(s,r) \in I_R'\times I_L'$. Let $M_{\mathrm{close}}$ and  $M_{\mathrm{far}}$ denote the close and far edges of $M$ respectively.  For a matching edge $(s,r) $, we denote it as a \emph{medium edge} if  the request $r$ is inside the interval $[\eps\Delta, (1-\eps\Delta)]$ and the server $s$ is inside the interval $I_L$ or $I_R$. We denote this set of edges as the \emph{medium edges} and denote it by $M_{\mathrm{med}}$. From the well-separated property of the input it is easy to see that $M = M_{\mathrm{far}}\cup M_{\mathrm{med}}\cup M_{\mathrm{close}}$.   A matching $M$ is \emph{$\eps$-well-aligned} if all the edges of $M_{\mathrm{close}}$ with both their endpoints inside $I_R'$ (resp. $I_L'$) are right (resp. left) edges. See Figure~\ref{fig:wam} for an example of $\eps$-well aligned matching of an $\eps$-well separated input instance. Any $\eps$-well-aligned matching of an $\eps$-well-separated input instance satisfies the following property.

\begin{figure}
\centering
\includegraphics[scale=0.9]{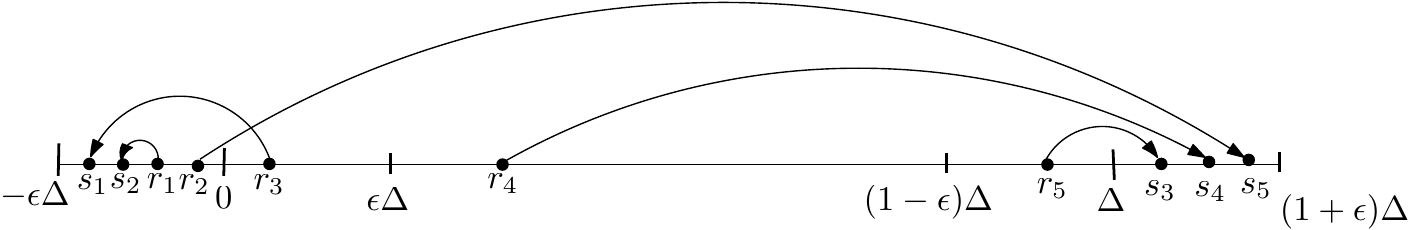}
\caption{All servers $S=\{s_1,s_2,s_3,s_4,s_5\} \subset [-\eps\Delta, 0]\cup[\Delta,(1+\eps)\Delta]$ and requests $R=\{r_1,r_2,r_3,r_4,r_5\} \subset [-\eps\Delta, (1+\eps)\Delta]$. So, $S\cup R$ is an $\eps$-well separated input instance.  The matching $M=\{(r_1,s_2),(r_2,s_5),(r_3,s_1), (r_4,s_4),(r_5,s_3)\}$ is partitioned into $M_{\mathrm{close}}=\{(r_1,s_2),(r_3,s_1),(r_5,s_3)\}$, $M_{\mathrm{far}}=\{(r_2,s_5)\}$ and $M_{\mathrm{med}}=\{(r_4,s_4)\}$. $M$ is $\eps$-well aligned since the edges of $M_{\mathrm{close}}$ in $[-\eps\Delta,(1+\eps)\Delta]$ are left edges (server is to the left of request) and those in $[(1-\eps)\Delta, (1+\eps)\Delta]$ are right edges (server is to the right of request).  }
\label{fig:wam}
\end{figure}

\begin{lemma}
\label{lem:wspc}
For any $0 \le \eps \le 1/8$, given an $\eps$-well-separated input $S$ and $R$ and an $\eps$-well-aligned matching $M$, let $M_{\opt}$ be the optimal matching of $S$ and $R$ and let $M_{\mathrm{close}}, M_{\mathrm{far}}$ and $M_{\mathrm{med}}$ be as defined above. Then,  $$\wt(M_\mathrm{close})+\wt(M_{\mathrm{med}}) \le \bigg(\frac{2}{\eps}+3\bigg)\wt(M_{\opt}) + \frac{4\eps}{1-2\eps} \wt(M_\mathrm{far}).$$
\end{lemma}  
\begin{proof}
Let $M_{\mathrm{close}}$ and $M_{\mathrm{far}}$, be the edges of $M$ that are in $I_{L}' \times I_{L}'$ (or $I_R'\times I_R'$) and $I_{R}'\times I_{L}'$ (or $I_L'\times I_R'$) respectively and $M_{\mathrm{med}}$ be the remaining edges of $M$. Let $S_{\mathrm{close}}$ and $R_{\mathrm{close}}$ be the servers and requests that participate in $M_{\mathrm{close}}$. Similarly, we define the sets $S_{\mathrm{far}}$ and $R_{\mathrm{far}}$ for $M_{\mathrm{far}}$ and the sets $S_{\mathrm{med}}$ and $R_{\mathrm{med}}$ for $M_{\mathrm{med}}$. Let $M_{\mathrm{close}}^{\opt}$ denote the optimal matching of $S_{\mathrm{close}}$ and $R_{\mathrm{close}}$ and let $M_{\mathrm{cf}}^{\opt}$ denote the optimal matching of $S_{\mathrm{far}}\cup S_{\mathrm{close}}$ with $R_{\mathrm{far}}\cup S_{\mathrm{close}}$. The following four claims will establish the lemma:
\begin{enumerate}
\item[(i)] $\wt(M_{\mathrm{close}}^{\opt})=\wt(M_{\mathrm{close}})$, i.e., $M_{\mathrm{close}}$ is an optimal matching of $S_{\mathrm{close}}$ and $R_{\mathrm{close}}$,
\item[(ii)] $\wt(M_{\mathrm{med}}) \le (1/\eps)\wt(M_{\opt})$,
\item[(iii)] $\wt(M_{\mathrm{cf}}^{\opt}) \le (1/\eps+3)\wt(M_{\opt})$,
and,
\item[(iv)] $\wt(M_{\mathrm{close}}^{\opt}) - 4\eps\Delta|M_{\mathrm{far}}| \le \wt(M_{\mathrm{cf}}^{\opt})$.

\end{enumerate} 
The matching $M_{\mathrm{close}}$ is a perfect matching of servers $S_{\mathrm{close}}$ and requests $R_{\mathrm{close}}$. Note that all edges of $M_{\mathrm{close}}$ are inside the interval $[-\eps\Delta, \eps\Delta]$ and $[(1-\eps)\Delta, (1+\eps)\Delta]$. Furthermore, those that are inside the interval  $[-\eps\Delta, \eps\Delta]$ are left edges and the edges that are inside $[(1-\eps)\Delta, (1+\eps)\Delta]$ are right edges. Therefore, $M_{\mathrm{close}}$ satisfies the precondition for (OPT) and so, $M_{\mathrm{close}}$ is an optimal matching of $S_{\mathrm{close}}$ and $R_{\mathrm{close}}$ implying (i).

To prove (ii), observe that any edge $(s,r)$ in $M_{\mathrm{med}}$ has the request $r$ inside the interval $[\eps\Delta, (1-\eps)\Delta]$. Therefore, the maximum length of any such edge is at most $\Delta$. On the other hand, let $s'$ be the match of $r$ in the optimal matching  $M_{\opt}$ . From the well-separated property, $s' \in [-\eps\Delta, 0]\cup [\Delta, (1+\eps)\Delta]$ and since $r \in [\eps\Delta, (1-\eps)\Delta]$, $\eps\Delta$ is a lower bound the length of $(s',r)$. Therefore, the cost of all the edges in $M_{\mathrm{med}}$ is bounded $(1/\eps) \wt(M_{\opt})$.\  

We prove (iii) as follows: Let $M_{\opt}$ be the optimal matching of $S$ and $R$. Note that every request in $R_{\mathrm{med}}$ is contained inside the interval $[\eps\Delta, (1-\eps)\Delta]$. Since all servers are in the interval $[0,-\eps\Delta]\cup [\Delta, (1+\eps)\Delta]$, every edge of $M_{\opt}$ that is incident on any vertex of $R_{\mathrm{med}}$ has a cost of at least $\eps\Delta$.   Initially set $M_{\mathrm{tmp}}$ to $M_{\opt}$. For every edge  $(s,r) \in M_{\mathrm{med}}$ , we remove points $s$ and $r$ and the edges  of $M_{\opt}$ incident on them from $M_{\mathrm{tmp}}$; note that the other end point of the edges of $M_{\mathrm{tmp}}$ incident on $s$ and $r$ can be any vertex of $R$ and $S$ including points from $R_{\mathrm{close}}\cup R_{\mathrm{far}}$ and $S_{\mathrm{close}}\cup S_{\mathrm{far}}$. After the removal of points, the vertex set of $M_{\mathrm{tmp}}$ is  $S_{\mathrm{close}}\cup S_{\mathrm{far}}$ and $R_{\mathrm{close}} \cup R_{\mathrm{far}}$. Removal of the edges can create at most $|M_{\mathrm{med}}|$ free vertices in $S_{\mathrm{close}}\cup S_{\mathrm{far}}$  with respect to $M_{\mathrm{tmp}}$. Similarly there are at most $|M_{\mathrm{med}}|$ free vertices in $R_{\mathrm{close}}\cup R_{\mathrm{far}}$ with respect to $M_{\mathrm{tmp}}$ .  We match these free nodes arbitrarily in $M_{\mathrm{tmp}}$ at a cost of at most $(1+2\eps)\Delta$ per edge. Therefore, the total cost of the matching $M_{\mathrm{tmp}}$ is at most $\wt(M_{\opt}) + |M_{\mathrm{med}}|(1+2\eps)\Delta$. For every $r \in R_{\mathrm{med}}$ the edge of $M_{\opt}$ incident on $r$ has a cost of at least $\eps\Delta$. Therefore, the cost of $M_{\opt}$ is at least $|M_{\mathrm{med}}|\eps\Delta$.  Combined, the new matching $M_{\mathrm{tmp}}$ matches $S_{\mathrm{close}}\cup S_{\mathrm{far}}$ to $R_{\mathrm{close}}\cup R_{\mathrm{far}}$ and has a cost at most $(1/\eps +3)\wt(M_{\opt})$ leading to (3).  

To prove (iv), let $|S_{\mathrm{close}}| = n_{\mathrm{close}}$. Consider the sequence $\sigma(S_{\mathrm{close}} \cup R_{\mathrm{close}})$ and let $K_j$ be the first $j$ points in this sequence and let $\kappa^j$ be the interval $[k^j, k^{j+1}]$. Let $\mathbb{K}$ be the set of all intervals $\kappa^j$ for $1\le j \le 2n_{\mathrm{close}}-1$. As described previously, we can express the cost of the optimal matching $M_{\mathrm{close}}^{\opt}$ as $\sum_j \mathrm{diff}(K_j)\mathscr{L}(\kappa^j)$. Let $j'$ be the largest index such that the coordinate value of the point $k^{j'}$ is at most $\eps\Delta$. By construction, all edges of $M_{\mathrm{close}}$ are contained inside $I_L'=[-\eps\Delta,\eps\Delta]$ and $I_R'=[(1-\eps)\Delta, (1+\eps)\Delta]$ and so,  $\mathrm{diff}(K_{j'})$ is $0$ and $\kappa^{j'}$ contains the interval $[\eps\Delta, (1-\eps)\Delta]$.
Therefore, we can express the cost of the optimal matching
 $M_{\mathrm{close}}^{\opt}$ as 
 $$\wt(M_{\mathrm{close}}^{\opt})= \sum_{\kappa^j \in \mathbb{K}\setminus \kappa^{j'}} \mathrm{diff}(K_j)\mathscr{L}(\kappa^j).$$ 
 Furthermore, since the intervals in $\mathbb{K}$ decompose the interval $[-\eps\Delta,(1+\eps)\Delta]$ and since $\kappa^{j'}$ contains $[\eps\Delta, (1-\eps)\Delta]$, we get 
\begin{eqnarray}\sum_{\kappa^j \in \mathbb{K}\setminus \kappa^{j'}}\mathscr{L}(\kappa^j) &\le& 4\eps\Delta.\label{eq:smallint}\end{eqnarray} 
 
When we add points of the set $S_{\mathrm{far}}$ and $R_{\mathrm{far}}$ to $S_{\mathrm{close}}\cup R_{\mathrm{close}}$, we further divide  every interval $\kappa^{j}$  into smaller  intervals say $\{\kappa^{j}_1,\kappa^{j}_2,\ldots, \kappa^{j}_{c_{j}}\}$. We can express the cost of the optimal matching
 $M_{\mathrm{close}}^{\opt}$ as 
 \begin{eqnarray}\wt(M_{\mathrm{close}}^{\opt})&=&\sum_{\kappa^j \in \mathbb{K}\setminus \kappa^{j'}} \mathrm{diff}(K_j)\mathscr{L}(\kappa^j)= \sum_{\kappa^j\in \mathbb{K}\setminus \kappa^{j'}} \sum_{i=1}^{c_j} \mathrm{diff}(K_j)\mathscr{L}(\kappa^j_i).\label{eq:closeopt}\end{eqnarray}

 Let $K_{j}^{i}$ be the subset of $S_{\mathrm{far}}\cup R_{\mathrm{far}} \cup S_{\mathrm{close}}\cup R_{\mathrm{close}}$ that are to the left of the interval $\kappa_{j}^{i}$ (including those that are on the left boundary of the interval). We can express the cost of the optimal matching $M_{\mathrm{cf}}^{\opt}$ as 
\begin{eqnarray}\wt(M_{\mathrm{cf}}^{\opt}) &=& \sum_{\kappa^j\in \mathbb{K}} \sum_{i=1}^{c_j} \mathrm{diff}(K_j^i)\mathscr{L}(\kappa^j_i)\ge \sum_{\kappa^j\in \mathbb{K}\setminus \kappa^{j'}} \sum_{i=1}^{c_j} \mathrm{diff}(K_j^i)\mathscr{L}(\kappa^j_i).\label{eq:cfcost}\end{eqnarray}The difference between the sets $K_j$ and $K_j^i$ is that $K_j^i$ contains points from $R_{\mathrm{far}}\cup S_{\mathrm{far}}$ and so, $\mathrm{diff}(K_j^i)\ge \mathrm{diff}(K_j) -|M_{\mathrm{far}}|$. Using this along with~\eqref{eq:cfcost} the above inequality, we get

$$\wt(M_{\mathrm{cf}}^{\opt}) )\ge \sum_{k_j\in \mathbb{K}\setminus \kappa_{j'}} \sum_{i=1}^{c_j} \mathrm{(diff}(K_{j})-|M_{\mathrm{far}}|)\mathscr{L}(\kappa_j^i) \ge$$ $$ \sum_{\kappa^j\in \mathbb{K}\setminus \kappa^{j'}} \sum_{i=1}^{c_j} \mathrm{diff}(K_j)\mathscr{L}(\kappa^j_i) - |M_{\mathrm{far}}|\sum_{\kappa^j \in \mathbb{K} \setminus \kappa^{j'}} \mathscr{L}(\kappa^j) \ge \wt(M_{\mathrm{close}}^{\opt}) - 4\eps\Delta|M_{\mathrm{far}}|,   $$     implying (iv). The last two inequalities follow from~\eqref{eq:smallint} and~\eqref{eq:closeopt}.\\  Since the length of every edge in $M_{\mathrm{far}}$ is at least $(1-2\eps)\Delta$, we can rewrite the above equation as 
$$\wt(M_{\mathrm{close}}^{\opt}) - \frac{4\eps}{1-2\eps}\wt(M_{\mathrm{far}}) \le \wt(M_{\mathrm{close}}^{\opt}) - \frac{4\eps}{1-2\eps}(1-2\eps)\Delta|M_{\mathrm{far}}| \le \wt(M_{\mathrm{cf}}^{\opt}).$$
The proof of this lemma follows from combining the above equation with (i), (ii) and (iii).

\end{proof} 

\subsection{Interpreting Dual Weights for the Line Metric}
\label{sec:dual}
Next, we interpret the dual weights and their changes during the RM-Algorithm for the line metric and derive some of its useful properties.

 \subparagraph*{Span of a request:} For any request $r \in R$, let $\ym{i}(r)$ be the largest dual weight that is assigned to $r$ in the first $i$ phases. The second step of phase $i$ does not increase the dual weights of requests, and so,  $\ym{i}(r)$ must be a dual weight assigned at the end of first step of some phase $j\le i$.  For any request $r$ and any phase $i$, the span of $r$, denoted by $span(r,i)$, is  an open interval that is centered at $r$ and has a length of $\frac{2\ym{i}(r)}{t}$, i.e., $span(r,i)=\left( r-\frac{\ym{i}(r)}{t}, r+\frac{\ym{i}(r)}{t}\right).$ We will refer to the closed interval  $[r-\frac{\ym{i}(r)}{t},r+\frac{\ym{i}(r)}{t}]$ with center $r$ and length $\frac{2\ym{i}(r)}{t}$ as $cspan(r,i)$.   

Intuitively, request $r$ may participate in one or more alternating trees in the first $i$ phases of the algorithm. The dual weight of every request that participates in an alternating tree, including $r$, increases. These increases reflect their combined search for a free server.  The region $span(r,i)$ represents the region swept by $r$ in its search for a free server  in the first $i$ phases. We show in Lemma~\ref{lem:emptyspan} that the span of any request does not contain a free server in its interior.  We show that, during the search, if the span of a request $r$ expands to include a free server $s \in S_F$ on its boundary, i.e., $s \in bd(span(r,i))$, then the algorithm would have found a minimum $t$-net-cost path and the search would stop. Therefore, the open interval $span(r,i)$ will remain empty.

\begin{lemma}
\label{lem:emptyspan}
For every $r$, $span(r,i) \cap S^i_F = \emptyset$.
\end{lemma}
\begin{proof}
For the sake of contradiction, we assume that the  span of a request $r$ contains a free server $s$, i.e., $s \in S_F^i \cap span(r,i)$. So, the distance between $r$ and $s$ is 
$|s-r| < \frac{\ym{i}(r)}{t},$ or,
\begin{equation}
\label{eq:lab}
\ym{i}(r) > t(|s-r|).
\end{equation}

Since $\ym{i}(r)$ is the largest dual weight assigned to $r$, there is a phase $j \le i$, when request $r$ is assigned this dual weight by the algorithm. Since the first step of the algorithm may only increase and the second step   may only decrease the dual weights of any request, we can assume that $y(r) $ was assigned the dual weight of $ \ym{i}(r)$  at the end of the first step of some phase $j$. Let $y(s)$ and $y(r)=\ym{i}(r)$ be the dual weights at  the end of the first step of phase $j$. From Invariant (I1), it follows that $\ym{i}(r) + y(s) \le t(|s-r|)$. From this and (\ref{eq:lab}), we have $y(s) < 0$. The free server $s$ has  $y(s)< 0$ contradicting invariant (I2).
\end{proof}

\begin{lemma}
\label{lem:included}
Let $(s,r)$ be any eligible edge at the end of the first step of phase $i$. Then,
$$y_{max}^i(r) \ge t|s-r|,$$
implying that $s\in cspan(r,i)$ and the edge $(s,r) \subset span(r,i)$.
\end{lemma}
\begin{proof}
An edge is eligible if it is in $\ms_{i-1}$ or if it satisfies~\eqref{eq:elig1}. Suppose $(s,r) \notin \ms_{i-1}$ and satisfies (\ref{eq:elig1}). In this case,
$y(s)+y(r)=t(|s-r|)$. From (I2), $y(s) \le 0$ and so,  
$y(r) \ge t(|s-r|)$, 
implying that $\ym{i}(r) \ge  t(|s-r|)$. 

For the case where $(s,r) \in \ms_{i-1}$, let $0<j< i$ be the largest index such that $(s,r) \in \ms_j$ and $(s,r)\not\in \ms_{j-1}$. Therefore, $(s,r)\in P_j\setminus \ms_{j-1}$. Since $P_j$ is an augmenting path with respect to $\ms_{j-1}$, every edge of $P_j\setminus \ms_{j-1}$ satisfies the eligibility condition (\ref{eq:elig2}) at the end of the first step of phase $j$ of the algorithm. For any vertex $v$, let $y'(v)$ be its dual weight after the end of the first step of phase $j$ of the algorithm. From (\ref{eq:elig2}), we have 
$y'(r) + y'(s) \le t|s-r|$. From (I2), since $y'(s) \le 0$, we have $y'(r) \ge t|s-r|$. By definition, $\ym{i}(r) \ge y'(r)$ and therefore $\ym{i}(r) \ge t|s-r|$.
\end{proof}

\subparagraph*{Search interval of a request:}Recollect that $B_i$ is the set of requests that participate in the alternating tree of phase $i$. In Lemma~\ref{lem:connected},  we show that $\bigcup_{r \in B_i} cspan(r, i)$ is a single closed interval. We define \emph{search interval} of $r_i$, denoted by  $sr(r_i)$, as the open interval $open(\bigcup_{r \in B_i}cspan(r,i))$. The search interval of a request represents the  region searched for a free server by all the requests of $B_i$. In Lemma~\ref{lem:connected}, we establish a useful property of search interval of a request. We  show that the search interval of $r_i$  does not contain any free servers of $S^i_F$ and the free server $s_i$ (the match of $r_i$ in the online matching)  is at the boundary of $sr(r_i)$. Since the search interval contains $r_i$, it follows that $s_i$ is either the closest free server to the left or the closest free server to the right of $r_i$. Using the fact that all requests of $B_i$ are connected to $r_i$ by an  path of eligible edges along with Lemma~\ref{lem:emptyspan} and Lemma~\ref{lem:included}, we get the proof for Lemma~\ref{lem:connected}. 
\begin{lemma}
\label{lem:connected}
After phase $i$, $\bigcup_{r \in B_i} cspan(r,i)$ is a single closed interval and so,  the search interval $sr(r_i)$ of any request $r_i$ is a single open interval. Furthermore,
\begin{itemize}
\item  All edges between $A_i$ and $B_i$ are inside the search interval of $r_i$, i.e.,$A_i\times B_i \subseteq  sr(r_i)$,
\item There are no free servers inside the search interval of $r_i$, i.e.,$S_F^i \cap sr(r_i) = \emptyset$, and, 
\item The server $s_i$ chosen by the algorithm is on the boundary of the search interval of $r_i$, i.e., $s_i \in bd(sr(r_i))$.
\end{itemize} 
\end{lemma}

\subparagraph*{Cumulative search region:}After phase $i$, the  \emph{cumulative search region} $\sr_{i}$ for the first $i$ requests $R_i$ is the union of the  individual search intervals $\sr_i=  (sr(r_1)\cup sr(r_2)\ldots \cup sr(r_i))$. Since the cumulative search region is the union of open intervals, it contains a set  of  interior-disjoint  open intervals.   Let $\sigma_i$ of $\sr_i$ be a  sequence of these interior-disjoint open intervals in the cumulative search region ordered from left to right, i.e., $\sigma_{i}=\langle \C_1^{i},\ldots,\C_k^{i}\rangle$, where $\C_j^{i}$ appears before $\C_l^i$ in the sequence if and only if the interval $\C_j^i $ is to the left of interval $\C_l^i
$.     In Lemma~\ref{lem:csrprop}, we establish properties of the cumulative search region. We show that every edge of the online matching $M_i$ and the offline matching $\ms_i$ is contained inside some interval of the sequence $\sigma_i$. We also show that there are no free servers inside any interval of $\sigma_i$, i.e.,  $S_F^i \cap \sr_{i} = \emptyset$.  
\begin{lemma}\label{lem:csrprop}
 After phase $i$ of the algorithm,
\begin{itemize}  
\item For every edge  $(s,r) \in \ms_i \cup M_i$ there exists an interval $\C \in \sigma_i$ such that  $(s,r)\subseteq \C$,
\item The set of free servers $S_F^i$ satisfies $S_F^i \cap \sr_i =\emptyset$, and there is an interval $\C$ in $\sigma_i$ such that the server $s_i\in bd(\C)$.
\end{itemize}    \end{lemma}

When a new request $r_{i+1}$ is processed by the algorithm, the search interval $sr(r_{i+1})$ is added to the cumulative search region, i.e., $\sr_{i+1}=\sr_{i}\cup sr(r_{i+1})$. Suppose $\langle \C_j^i,\ldots, \C_l^i\rangle$ intersects $sr(r_{i+1})$, then $\sr_{i+1}$ will have a single interval that contains intervals $\langle \C_j^i,\ldots, \C_l^i\rangle$ and $sr(r_{i+1})$. From this description, an easy observation follows: 

\begin{itemize} 
\item[(O1)] Given two intervals $\C$ and $\C'$ in cumulative search regions $\sigma_i$ and $\sigma_j$ respectively,  with $j > i$ , then either  $\C\cap \C' = \emptyset$ or $\C \subseteq \C'$.
\end{itemize}

\subparagraph*{Matchings in cumulative search regions:}From the property of  cumulative search region established in Lemma~\ref{lem:csrprop},  every edge of the online matching $M_i$ and the offline matching $\ms_i$ is contained inside some interval of the sequence $\sigma_i$. We denote the edges of online and offline matching that appear in an interval $\C$ of $\sigma_i$ by $M_{\C}$ and $\ms_{\C}$ respectively. Therefore, $M_i = \bigcup_{\C\in \sigma_i}M_{\C}$ and $\ms_i = \bigcup_{\C\in \sigma_i}\ms_{\C}$.   Note that $M_\C$ and $\ms_\C$ match the same set of servers and requests. Let this set of servers and requests be denoted by $S_{\C}$ and $R_\C$ respectively.

Consider any sequence of interior-disjoint intervals $\langle \C_{j_1}^{i_1},\ldots, \C_{j_k}^{i_k}\rangle $, where each interval $\C_{j_{l}}^{i_l}$  appears in the  cumulative search region $\sigma_{i_l}$. Note that every interval in this set can appear in after the execution of a different phase. In Lemma~\ref{lem:costbnd}, we show that  $\sum_{l=1}^k\wt(\ms_{\C_{j_l}^{i_l}})\le t\wt(M_{\opt}) $, i.e., the  total cost of the subset of offline matching edges that are contained inside all of these  disjoint intervals is within a factor of $t$ times the cost of the optimal matching. This property, therefore, relates the cost of subsets of offline matchings, each of which appear at different phases to the optimal cost.

\begin{lemma}
\label{lem:costbnd}
For any $i,j$, let $\C_j^i$ be the $j$th interval in the sequence $\sigma_i$. Suppose we are given intervals $\C^{i_1}_{j_1}, \C^{i_2}_{j_2}\ldots, \C^{i_k}_{j_k}$ such that no two of these intervals have any mutual intersection. Then $\wt(\ms_{\C^{i_{1}}_{j_{1}} }) +\wt(\ms_{\C^{i_2}_{j_2}}) +\ldots + \wt(\ms_{\C^{i_k}_{j_k}}) \le t \wt(M_{\opt})$.
\end{lemma}
\begin{proof}
Let $S_{\C_j^i}$ and $R_{\C_j^i}$ be the servers and requests belonging to the interval $\C_j^i$.  Let $y^i(\cdot)$ be the dual weights of vertices in $S\cup R$\ at the end of phase $i$. To prove the lemma, we will assign a dual weight $y(\cdot)$ for every vertex such that  this assignment along with the matching $\M=\bigcup_{l=1}^k\ms_{\C_{j_l}^{i_l}}$ is a $t$-feasible matching that satisfy the feasibility conditions (\ref{eq:feas1}) and (\ref{eq:feas2}).  The dual assignment is made as follows
 \begin{enumerate}
\item 
Every server belonging to $s \in S \setminus  \bigcup_{l=1}^k S_{\C_{j_l}^{i_l}}$ and every request belonging to $r\in R \setminus \bigcup_{l=1}^k R_{\C_{j_l}^{i_l}}$ is assigned a dual weight of $0$, i.e., $y(r)\leftarrow 0, y(s) \leftarrow 0$. 
\item For each $1\le l\le k$, every server $s \in S_{\C_{j_l}^{i_l}} $ and every request $r\in R_{\C_{j_l}^{i_l}}$ are assigned a dual weight which is equal to the dual weight of s  and $r$ respectively after the completion of  phase $i_{l}$, i.e., $y(s) \leftarrow y^{i_l}(s)$, $y(r)\leftarrow y^{i_l}(r)$.
\end{enumerate} 
The matching $\M$ along with this dual assignment\ $y(\cdot)$ for vertices form $t$-feasible matching. Notice that, for every edge   $(s,r)\in \M$, the edge $(s,r)$ will be contained in one of the intervals;  let $\C^{i_l}_{j_l}$  be this interval. Since both end points of the matching edge $(s,r)$ are contained in the same interval $\C_{j_{l}}^{i_l}$\ after phase $i_{l}$, from our dual assignment, we have $y(s)=y^{i_l}(s)$ and $y(r)=y^{i_l}(r)$.  Since $(s,r)$ belongs to the offline matching at the end of phase $i_{l}$, i.e., $\ms_{\C_{j_l}^{i_l}} \subseteq \ms_{i_{l}}$ and since $\ms_{i_l}$ and $y^{i_l}(\cdot)$ form a $t$-feasible matching (from (I1)), we have, $y(s)+y(r) = y^{i_l}(s)+y^{i_l}(r) = |s-r|$. Therefore, every edge in the matching satisfies (\ref{eq:feas2}). For any edge $(s,r)$ that is not in $\M$,   let $r \in R_{\C_{j_l}^{i_l}}$. Then, there are two possibilities. Either $s \in S_{\C_{j_l}^{i_l}}$ or $s \notin S_{\C_{j_l}^{i_l}}$. We consider these cases separately. First, $s\in S_{\C_{j_l}^{i_l}}$. In this case, we can simply argue that both $r$ and $s$  have a dual weight $y^{i_l}(s)$ and $y^{i_l}(r)$ respectively. Since the dual weights after phase $i_l$ were $t$-feasible, from (\ref{eq:feas1}), we have, $y(s)+y(r) = y^{i_l}(s)+y^{i_l}(r) \le t|s-r|$. The second possibility is that $s \not\in S_{\C_{j_l}^{i_l}}$. Let $s\in S_{\C_{j_{l'}}^{i_{l'}}}$.  In this case, due to disjointness of $\C^{i_l}_{j_l}$ and $\C^{i_{l'}}_{j_{l'}}$, the server $s$ will not be contained in the span of $r$, i.e., $s\notin span(r,i_l)$. This implies, $y(r) =y^{i_l}(r) \le \ym{i_l}(r)$. Since $s$ is not in $span(r,i_l)$ (which is a ball of radius $\ym{i_l}(r)$, we have $|s-r| \ge \ym{i_l}(r)/t$. Therefore,  $y(r)=y^{i_l}(r) \le \ym{i_l}(r)\le t|s-r|$. From (I2), we have that $y^{i_{l'}}(s)\le 0$. Therefore, $y(r)+y(s) \le t|s-r|$, satisfying (\ref{eq:feas1}).  

Since, $\M$ is a $t$-feasible matching, next, we show that the cost of $\M$ is no more than $t\wt(M_{\opt})$. By our dual assignment, every unmatched server and request with respect to the matching $\M$ has a dual weight of $0$. For any edge $(s,r)\in \M$, the sum of the dual weights of $s$ and $r$ is exactly equal to $|s-r|$ (due to $t$-feasibility of dual weights $y(\cdot)$ and the matching $\M$). Therefore, the sum of all the dual weights is exactly equal to the cost of $\M$, i.e., $\sum_{v \in S\cup R} y(v) = \wt(\M)$. Next, consider the edges of the optimal matching $M_{\opt}$. For each such edge of $M_\opt$, since the dual weights $y(\cdot)$ satisfies the $t$-feasibility condition $y(s)+y(r) \le t|s-r|$, we have that the sum of the dual weights is $\sum_{v \in S\cup R} y(v) \le t\wt(M_{\opt}$). From this, we deduce that $\wt(\M)\le t\wt(M_{\opt})$ as desired.  
\end{proof}

\section{Analysis for the Line Metric}
\label{sec:1danalysis}
  
For each interval of any cumulative search region $\sigma_i$, we assign it a level between $0$ and $O(t\log nt)$  based on the length of the interval. We also partition the edges of the online matching into $O(t \log nt)$ levels.
\subparagraph*{Level of an Interval:} For any $0 <  i \le n$,  we assign a \emph{level} to every  interval of any cumulative search region $\sigma_i$.  The \emph{level} of any interval $\C_j^i\in \sigma_i$   , denoted by $\level(\C_j^i)$, is $k$ if   

$$(\appr)^k (\wt(M_{\opt})/n) \le \ints(\C_j^i) \le (\appr)^{k+1}(\wt(M_{\opt})/n).$$
Here $\ints(\C_j^i)$ is the length of the interval $\C_j^i$ . All intervals whose length $\ints(\C_j^i) \le \wt(M_{\opt})/n$ is assigned a level $0$. 

\subparagraph*{Level of an online edge:} For any request $r_{i}$ and its match  $s_{i}$ in $M_{i}$, let $r_{i}$ and $s_{i}$ be contained in an interval $\C'\in \sigma_{i} $. Then, the level of this edge  $(r_{i},s_{i})$, denoted by $\level(r_{i})$, is given by the level of the interval $\C'$, i.e., $\level(r_{i}) = \level(\C')$.

\begin{lemma}
\label{lem:numlevels}
The largest level of any online edge is $O(t\log (nt))$.
\end{lemma}

\subparagraph*{Tracking the evolution of cumulative search region:} The cumulative search region $\sr_{i}$ after any phase $i$ will include disjoint intervals from the sequence $\sigma_i$. 
 When we process a new request $r_{i+1}$, the cumulative search region is updated to include the open interval $sr(r_{i+1})$.  This may combine a contiguous sub-sequence of intervals $\Gamma_{i+1}=\langle \C_j^i, \C_{j+1}^i,\ldots, \C_l^i\rangle$  of $\sigma_i$ along with the interval $sr(r_{i+1})$ to form a single open interval $\C_j^{i+1}$ in $\sigma_{i+1}$. We refer to the sequence of intervals   $\Gamma_{i+1}$   as the \emph{predecessors} of $\C_j^{i+1}$ and denote $i+1$ as the \emph{birth phase} of $\C_j^{i+1}$. For each interval $\C$ in $\Gamma_{i+1}$, we define its \emph{successor}  to be    $\C_j^{i+1}$  and denote $i+1$ as the \emph{death phase} of $\C$.
 Suppose $\C_j^{i+1}$ is a level $k$ interval. Then, it is easy to see that there is at most one level $k$ interval  in $\Gamma_{i+1}$.
\begin{lemma}
\label{lem:onelevelk}
For $k \ge 1$ and any level $k$ interval $\C_{j}^{i+1}$, there is at most one level $k$ interval in the predecessor set $\Gamma_{i+1}$.
\end{lemma}
\begin{proof}
For the sake of contradiction, suppose there are at least two  interior-disjoint level $k$ intervals $\C$ and $\C'$ . Since  $\C_{j}^{i+1}$ contains both $\C$ and $\C'$, its length is at least $2(\appr)^k (\wt(M_{\opt})/n)$ contradicting the fact that $\C_{j}^{i+1}$ is  a level $k$ interval.   
\end{proof}
 Recollect
that $M_{\C_j^{i+1}}$ and $\ms_{\C_j^{i+1}}$ are edges of the online and offline matching  inside $\C_j^{i+1}$that match the servers of $S_{\C_j^{i+1}}$ to requests of $R_{\C_j^{i+1}}$.  Let  $M_{\Gamma_{i+1}} = \bigcup_{\C \in \Gamma_{i+1}}M_{\C}$ be the edges of the  online matching contained inside the predecessors $\Gamma_{i+1}$ of $\C_{j}^{i+1}$.  By construction,  the matchings $M_{\C_j^{i+1}}$ and  $M_{\Gamma_{i+1}}$ only differ in the edge $(s_{i+1},r_{i+1})$.  So, $M_{\Gamma_{i+1}}$ match servers in $S_{\C^{i+1}_j}\setminus \{s_{i+1}\}$ to requests in $R_{\C_j^{i+1}}\setminus\{r_{i+1}\}$.

\subparagraph*{Maximal and minimal level $k$ interval:} A level $k$ interval $\C$ is \emph{maximal} if 
\begin{itemize}
\item $\C$ is an interval in the cumulative search region  after all the $n$ requests have been processed, i.e., $\C \in\sigma_n$ or,
\item if  the successor  $\C'$  of $\C$ has  $\level(\C) >k$.
\end{itemize}
 A level $k$ interval $\C_j^{i+1}$ is a \emph{minimal level $k$ interval} if none of its predecessors in $\Gamma_{i+1}$ are of level $k$.

  Next, we will describe a sequence of nested level $k$ intervals that  capture the evolution of a minimal level $k$ interval into a maximal level $k$ interval. For any level $k$ interval $\C$, we define a sequence of level $k$ intervals $\evol_{\C }$ along with a set $\comp(\C)$  of intervals of  level strictly less than $k$ in  a recursive way as follows: Initially $\evol_{\C}=\emptyset$ and $\comp(\C)=\emptyset$.   Suppose phase $i$ is the birth phase of $\C$. If all  the intervals in the predecessor  $\Gamma_{i}$  have a level less than $k$, then $\C$ is a minimal level $k$ interval. We add $\C$ to to front of the sequence $\evol_{\C}$ and add the predecessors of $\C$, $\Gamma_i$, to the set  $\comp(\C)$ and return $\evol_{\C}$ and $\comp(\C)$. Otherwise, from Lemma~\ref{lem:onelevelk}, there  is exactly one level $k$ interval $\C' \in \Gamma_i$ and all other intervals have a level strictly less than $k$. In this case, we compute  $\evol_{\C'}$ and $\comp(\C')$ recursively. We set $\evol_{\C}$ to the sequence $\evol_{\C'}$ followed by $\C$. We add all intervals in $\Gamma_i \setminus \C'$ along with the intervals of $\comp(\C')$ to $\comp(\C)$ and return $\comp(\C)$ and $\evol_{\C}$.   

For any maximal level $k$ interval $\C$, consider the sequence $\evol_{\C}= \langle \C_1,\ldots, \C_l\rangle$ and the set of disjoint intervals $\comp(\C)$ . From the construction of $\evol_{\C}$, the following observations follow in a straight-forward way:  
\begin{enumerate}[label=(E\arabic*)]
\item All intervals $\C' \in \evol_{\C}$ are level $k$ intervals. The first and the last interval, i.e., $\C_1$ and $\C_l=\C$ in $\evol_{\C}$ are minimal and maximal level $k$ intervals respectively, 
\item The intervals in $\evol_{\C}$ are nested, i.e., $\C_{i} \subseteq \C_{i+1}$ for all $1\le i \le l-1$.
\end{enumerate}

Let $\Mk_{\C}$ be the set of all level $k$ edges of the online matching that are contained inside $\C$. Let $\S_{\C}$ and $\R_{\C}$ be the servers and request that are matched by $\Mk_{\C}$. From the construction of $\comp(\C)$, the following properties follow in a straight-forward way:
\begin{enumerate}[label=C\arabic*]
\item[(C1)] $\comp(\C)$ is a set of disjoint intervals each of which is contained inside $\C$,
\item[(C2)] Every point $s \in S_{\C}\setminus \S_{\C}$ and $r \in R_{\C}\setminus \S_{\C}$ is contained inside some interval from the set $\comp(\C)$. 
\end{enumerate}

Let $I=\{I_1,\ldots, I_{l_k}\}$ be the set of all maximal level $k$ intervals. In the following lemma, we show that  $I$ is a set of pairwise interior-disjoint intervals. Furthermore,   the matchings $\Mk_{I_1}, \Mk_{I_2}, \ldots, \Mk_{I_{l_{k}}}$ partitions  all the level $k$ edges of the online matching $M$.

\begin{lemma}
Let $I=\{I_1,\ldots, I_{l_k}\}$ be the set of all maximal level $k$ intervals. Then, $I$ is a set of  interior-disjoint intervals. Furthermore,   the matchings $\Mk_{I_1}, \Mk_{I_2}, \ldots, \Mk_{I_{l_{k}}}$ partitions  all the level $k$ edges of the online matching $M$. 
\end{lemma}
\begin{proof}
Any two maximal level $k$ intervals $\C$ and $\C'$ will be pairwise-disjoint. If not, from (O1), without loss of generality, $\C$ will be contained inside $\C'$ implying that $\C$ is not a maximal level $k$ interval.

Consider any level $k$ edge $(s_i,r_i)$ and let the interval $\C^i_j$ contain the edge $(s_i,r_i)$.  Set $\C'' \leftarrow \C_{j}^i$. If the successor interval $C$ of $\C''$ is also level $k$, then we set $\C''\leftarrow C$ and repeat this step. Otherwise, if the successor interval $C$ is of a level higher than $k$ or if $C \in \sigma_n$, then we set  $\C \leftarrow \C''$ as the maximal level $k$ interval. By construction, $\C^{i}_j$ is in $\evol_{\C}$ and $(s_i,r_i) \in \Mk_{\C}$. Therefore, every level $k$ edge of the online matching is contained inside the maximal level $k$ interval $\C$.  \end{proof}

\begin{lemma}
\label{lem:wspconline} 
For any level $k\ge1$ interval $\C$, consider the sequence $\evol_{\C}$.   Let $\eps = \frac{1}{32t}$ and $t=3$. Then, $\R_{\C}$ and $\S_{\C}$ is an $\eps$-well separated input and the matching $\Mk_{\C}$ is an $\eps$-well-aligned matching of $\S_{\C}$ and $\R_{\C}$. Furthermore, each far edge in the $\eps$-well aligned matching $\Mk_{\C}$ is also a  long edge of the online matching. 
\end{lemma}
\begin{proof}
 Let $\C'$ be the minimal level $k$ interval of $\evol_{\C}$ (i.e., the first  interval in the sequence $\evol_{\C}$) and let $\Delta$ be the length of $\C'$, i.e., $\Delta=\ints(\C')$. Without loss of generality, we assume that $\C' $ is the open interval $ (0,\Delta)$. From (E2), every interval $\C_j$ in $\evol_{\C}$ contains $\C'$. Furthermore, $\C_j$ is inside $[-\eps\Delta, (1+\eps)\Delta]$. Suppose $\C_{j}$ is not contained inside $[-\eps\Delta, (1+\eps)\Delta]$, then the length $\ints(\C_j) \ge (1+\eps)^{k+1}(\wt(M_{\opt})/n)$. This contradicts the fact that $\C_j$ is level $k$ interval. 

 Let $i$ be the birth phase of $\C'$. From Lemma~\ref{lem:csrprop}, $s_i \in bd(\C')$ . From Lemma~\ref{lem:csrprop},  there are no servers of $S_F^i$ inside $\C'$. Since   $\S_{\C}\subseteq S_F^i$, all servers of $\S_{\C}$ are in $[-\eps\Delta, 0]$ and $[\Delta, (1+\eps)\Delta]$. From these facts, we conclude that $\R_{\C}$ and $\S_{\C}$ together form an $\eps$-well separated input. 

Next, we show that $\Mk_{\C}$ is an $\eps$-well-aligned matching of $\S_{\C}$ and $\R_{\C}$. It suffices if we show that any edge of $\Mk_{\C}$ contained in the interval $I_R'=[(1-\eps)\Delta, (1+\eps)\Delta]$ is a right edge. For the sake of contradiction, suppose there is a left edge $(s_{l},r_{l})$ in $I_R'$ Since $s_l \in [\Delta,(1+\eps)\Delta]$,  we have $\Delta< s_{l} < r_{l} \le (1+\eps)\Delta$.  $(s_{l},r_{l})$ is a level $k$ edge and therefore, there is a level $k$ interval $\tilde{\C}$ in $\evol_{\C}$ for which phase $l$ is the birth phase. By construction $\tilde{\C}$ contains both $[0,\Delta]$ and $r_{l}$ and so, $s_{l}$ is in the interior of $\tilde{\C}$. This contradicts the property (from  Lemma~\ref{lem:csrprop}) of $\sr_{l}$ that $s_l$ is on the boundary of $\tilde{\C}$. Therefore, $(s_l,r_l)$ is a right edge. A symmetric argument can be used to show that the edges of $\Mk_{\C}$ in the interval $I_{L}'=[-\eps\Delta, \eps\Delta]$ are left edges. Thus we conclude that $\Mk_{\C}$ is an $\eps$-well aligned matching of an $\eps$-well separated input.

Finally, consider any far edge $(s_l,r_l)$ of $\Mk_{\C}$, i.e., $(s_l,r_l) \in  I_{L}'\times I_R'$ (a symmetric argument works when $(s_l,r_l) \in I_R'\times I_L'$). For the sake of contradiction, suppose $(s_l,r_l)$ is a short edge in the online matching. We can bound the length of the edge $(s_l,r_l)$ from above by the length of the augmenting path $P_l$ computed in phase $l$ and we have  $\ell(P_l) \ge \ints(s_l,r_l) \ge (1-2\eps)\Delta$. From the definition of short edge, we have $\phi_t(P_l) \ge \frac{(t-1)}{4}\ell(P_l) \ge \frac{t-1}{4}(1-2\eps)\Delta$. For  our choice of $\eps=\frac{1}{32t}$ and $t=3$, we get $\phi_t(P_l)/t \ge 3\eps\Delta$ , the span $span(r_l, l) \supseteq [r_l-3\eps\Delta, r_l+3\eps\Delta]$. Since $r_l\in[(1-\eps)\Delta,(1+\eps)\Delta]$, it follows that the right end point of $span(r_l,l)$ is at least $(1+2\eps)\Delta$. Since span of $r_l$ is contained in the search interval $sr_{r_l}$ and also the cumulative search range $\sr_l$, the interval $\C'$ whose birth phase is $l$ contains the point $(1+2\eps)\Delta$. From this, it follows that $\C'$ is not inside $[-\eps\Delta, (1+\eps)\Delta]$ contradicting the fact that has already been established that every interval of $\evol_{\C}$ is inside $[-\eps\Delta, (1+\eps)\Delta]$. 
\end{proof}

\noindent

For any  level $k$ interval $\C$, let  $\Mk^{\opt}_{\C}$ be the optimal matching of $\S_{\C}$ and $\R_{\C}$. Suppose $\{I_1,\ldots, I_{l_k}\}$ be the set of all maximal level $k$ intervals.  Let $\ms_{\comp(I_{j})}$ be the union (over all intervals of $\comp(I_{j})$), the offline matching contained inside an interval of $\comp(I_{j})$. By (C2) and Lemma~\ref{lem:csrprop},  $\ms_{\comp(I_{j})}$ matches   servers in $S_{I_{j}}\setminus \S_{I_{j}}$ to requests $R_{I_{j}}\setminus \R_{I_{j}}$. The symmetric difference of $\bigcup_{I_j \in I}\ms_{\comp(I_j)}$and $\bigcup_{I_j \in I}\ms_{I_j}$consists of augmenting paths that match $\bigcup_{I_j \in I} \S_{I_j}$ to $\bigcup_{I_j\in I} \R_{I_j}$. Note that these are precisely the points that are matched by level $k$ online edges. From Lemma~\ref{lem:costbnd},  $w(\bigcup_{I_j \in I}\ms_{\comp(I_j)})\le t\wt(M_{\opt}) $and $\wt(\bigcup_{I_j \in I}\ms_{I_j})\le t\wt(M_{\opt})$ leading to the following lemma. 
\begin{lemma}
\label{lem:optcost}
For any $k > 0$, let $\{I_1\ldots, I_{l_{k}}\}$ be the set of all maximal level $k$ intervals. Then, $\sum_{j=1}^{l_k} \wt(\Mk^{\opt}_{I_j}) \le 2t\wt(M_{\opt})$.
\end{lemma}
\begin{proof}
Consider any interval $I_j$. Since $\ms_{I_j}$ is a perfect matching of $S_{I_j}$ and $R_{I_j}$ where as $\ms_{\comp(I_j)}$ is a perfect matching  of $S_{I_j}\setminus \S_{I_j}$ with vertices of $R_{I_{j}} \setminus \R_{I_j}$. Therefore, $\ms_{I_j} \oplus \ms_{\comp(I_j)}$ consists of exactly  $p=|\S_{I_j}| = |\R_{I_j}|$ vertex disjoint augmenting paths $\EuScript{P}=\{P^1,\ldots, P^p\}$ with respect to the matching $\ms_{\comp(I_j)}$. Each $P^i$  goes from some vertex in  $s \in \S_{I_j}$ to a vertex   $r\in \R_{I_j}$. Therefore, the set of augmenting paths induce a matching of $\S_{I_j}$ and $\R_{I_j}$.  Let $M'$ be this matching. From the metric property, $|s-r| \le \ell(P^i)$.  Therefore, the cost of $M'$
$$\wt(\Mk^{\opt}_{I_j}) \le \wt(M') = \sum_{(s,r) \in M'} | s-r| \le \sum_{P^i \in \EuScript{P}} \ell(P^i) \le \wt(\ms_{I_j}) + \wt(\ms_{\comp(I_j)}).$$ Adding the above over all maximal level $k$ intervals, 
\begin{eqnarray}\label{eq:costrel1}\sum_{j=1}^{l_k}\wt(\Mk^{\opt}_{I_j}) &\le& \sum_{j=1}^{l_k}\wt(\ms_{I_j}) + \sum_{j=1}^{l_k}\wt(\ms_{\comp(I_j)}).\end{eqnarray}    Since all the maximal level $k$ intervals are pairwise interior-disjoint, from Lemma~\ref{lem:costbnd}, we have  
\begin{eqnarray}\label{eq:costrel2}\sum_{j=1}^{l_k}\wt(\ms_{I_j})&\le& t\wt(M_{\opt}).\end{eqnarray} Since all intervals in $\comp(I_j)$  are interior-disjoint and contained inside $I_j$, all the intervals in $\bigcup_{j=1}^{l_{k}} \comp(I_j)$ are pairwise interior-disjoint and from Lemma~\ref{lem:costbnd},  \begin{eqnarray}\label{eq:costrel3}\sum_{j=1}^{l_k}\wt(\ms_{\comp(I_j)}) \le t\wt(M_{\opt}).\end{eqnarray} The lemma follows by combining~\eqref{eq:costrel1},\eqref{eq:costrel2} and~\eqref{eq:costrel3}.
\end{proof}
For $k > 0$, consider any maximal level $k$  interval $I_j$. In Lemma~\ref{lem:wspconline}, for $\eps > 1/32t$ and $t=3$, we show that the matching $\Mk_{I_j}$ is an $\eps$-well-aligned matching of an $\eps$-well-separated input instance. Additionally, we also show in this lemma  that every far edge of this well-aligned matching is also a long edge of the online matching. 
Let $\Mk_{I_j}^\mathrm{far}$, $\Mk_{I_j}^\mathrm{close}$ and $\Mk_{I_j}^\mathrm{med}$ denote the far,  close and medium edges of $\Mk_{I_j}$. Recollect that  $\Mk_{I_j}^\mathrm{far}\cup \Mk_{I_j}^\mathrm{close} \cup \Mk_{I_j}^\mathrm{med} = \Mk_{I_j}$. From Lemma~\ref{lem:wspc}, we know that 
$$\wt(\Mk_{I_j}^{\mathrm{close}})+\wt(\Mk_{I_j}^{\mathrm{med}}) \le (2/\eps + 3)\wt( \Mk_{I_j}^{\opt}) + \frac{4\eps}{1-2\eps} \wt(\Mk_{I_j}^{\mathrm{far}}).$$ 

Adding the above inequality over all level $k$ intervals and setting $\eps=\frac{1}{32t}$ and $t=3$, we get

$$\sum_{j=1}^{l_k}(\wt(\Mk_{I_j}^{\mathrm{close}})+\wt(\Mk_{I_j}^{\mathrm{med}})) \le (2/\eps + 3)\sum_{j=1}^{l_k}\wt( \Mk_{I_j}^{\opt}) + \frac{4\eps}{1-2\eps} \sum_{j=1}^{l_k}\wt(\Mk_{I_j}^{\mathrm{far}}).$$
Adding the above equation for all levels $k >0$ and adding $\wt(M_{\opt})$ to the RHS and LHS, we get 

\begin{equation*}\wt(M_{\opt})+\sum_{k=1}^{O(\log n)}\sum_{j=1}^{l_k}(\wt(\Mk_{I_j}^{\mathrm{close}})+\wt(\Mk_{I_j}^{\mathrm{med}})) \le\end{equation*}
\begin{equation}\label{eq:longshort1} \wt(M_{\opt})+(\frac{2}{\eps}+3) \sum_{j=1}^{l_k}\wt( \Mk_{I_j}^{\opt}) + \frac{2}{47} \sum_{k=1}^{O(\log n)}\sum_{j=1}^{l_k}\wt(\Mk_{I_j}^{\mathrm{far}}).\end{equation}  Note that every short edge of the online matching with level $k > 0$ will appear, for some maximal level $k$ interval $I_j \in \{I_1,\ldots, I_{l_k}\}$, in $\Mk_{I_j}^{\mathrm{close}}$ or $\Mk_{I_j}^{\mathrm{med}}$. By construction, the short edges of level $0$ have a length of at most $\wt(M_{\opt})/n$ and there are at most $n$ such edges. Recollect that $M_H$ is the set of short online edges. Therefore, 
$$\wt(M_H) \le \wt(M_{\opt})+\sum_{k=1}^{O(\log n)}\sum_{j=1}^{l_k}(\wt(\Mk_{I_j}^{\mathrm{close}})+\wt(\Mk_{I_j}^{\mathrm{med}})). $$
Every edge in $\Mk_{I_j}^{\mathrm{far}}$ is only a long edge of the online matching. Recollect that $M_L$ is the set of long edges of the online matching. We can, therefore, rewrite the~\eqref{eq:longshort1} as $$\wt(M_H) \le O(1/\eps)\sum_{k=1}^{O(\log n)}\sum_{j=1}^{l_k}\wt( \Mk_{I_j}^{\opt}) + \frac{2}{47} \ \wt(M_L),$$ 
  \begin{eqnarray*}
\wt(M_H) &\le& O(\log n) \wt(M_{\opt}) + \frac{2}{47}\wt(M_L).
\label{eq:count} 
\end{eqnarray*}

The above inequality follows from  Lemma~\ref{lem:optcost} and the fact that every online edge appears in exactly one of the maximal level $k$ intervals for some level $k$.  In Lemma~\ref{lem:shortcost},  by setting $t=3$, we get $\wt(M_H) \ge \wt(M)/6$ and $\wt(M_L) \le (5/6)\wt(M) $, and 
\begin{eqnarray*}
\wt(M_H) - \frac{2}{47}\wt(M_L) & \le & O(\log n)\wt(M_{\opt}),\\
\frac{\wt(M)}{6} - \frac{2}{47}\times \frac{5}{6}\wt(M) &\le & O(\log n)\wt(M_{\opt}),\\ \frac{\wt(M)}{\wt{(M_{\opt}})} &\le& O(\log n).\\
\end{eqnarray*}



\bibliography{onlinematch}
\end{document}